\documentclass[preprint,12pt]{article}

\usepackage{amssymb}

\usepackage{graphicx}
\usepackage[english]{babel}

\usepackage{float}
\usepackage{enumitem}
\usepackage{amsmath}
\usepackage{cleveref}
\usepackage{multirow}
\usepackage{tabularx}
\usepackage{xcolor}
\usepackage{url}
\usepackage{mathtools}

\newtheorem{theorem}{Theorem}[section]
\newtheorem{corollary}{Corollary}[section]
\newtheorem{proposition}{Proposition}[section]
\newtheorem{lemma}{Lemma}[section]
\newtheorem{example}{Example}[section]
\newtheorem{remark}{Remark}[section]

\newenvironment{proof}[1][Proof.]{\vspace{0.5em}\textbf{#1} }{\
\hfill$\square$\medskip}

\newcommand{\bv}{\mathbf{v}}

\newcommand{\Z}{\mathbb{Z}}
\newcommand{\zero}{{\mathbf{0}}}
\newcommand{\one}{{\mathbf{1}}}

\newcommand{\pmo}{{\mathbf{p-1}}}

\newcommand{\C}{{\cal C}}
\newcommand{\N}{{\mathbb{N}}}

\newcommand{\rank}{\text{rank}}
\newcommand{\kernel}{\text{ker}}
\newcommand{\ord}{\operatorname{ord}}

\newcommand{\cS}{{\cal S}}
\newcommand{\cA}{{\cal A}}

\begin{document}




\title{Equivalences among $\Z_{p^s}$-linear Generalized\\ Hadamard Codes
\thanks{This work has been partially supported by the Spanish MINECO under Grant PID2019-104664GB-I00
(AEI / 10.13039/501100011033) and by Catalan AGAUR scholarship 2020 FI SDUR 00475.}
}

\author{Dipak K. Bhunia, Cristina Fern\'andez-C\'ordoba,\\ Carlos Vela, Merc\`e Villanueva
\thanks{D. K. Bhunia, C. Fern\'andez-C\'ordoba, and Merc\`e Villanueva are with the Department of Information and Communications
Engineering, Universitat Aut\`{o}noma de Barcelona, 08193 Cerdanyola del Vall\`{e}s, Spain; and C. Vela is with the Department of Mathematics, University of Aveiro, 3810-197 Aveiro, Portugal.}
}

\maketitle


\begin{abstract}
The $\Z_{p^s}$-additive codes of length $n$ are subgroups of $\Z_{p^s}^n$, and can be seen as a generalization of linear
codes over $\Z_2$, $\Z_4$, or $\Z_{2^s}$ in general. A $\Z_{p^s}$-linear generalized Hadamard (GH) code is a GH code over $\Z_p$ which is the image of a $\Z_{p^s}$-additive code by a generalized Gray map. A partial classification of these codes by using the dimension of the kernel is known. In this paper,
we establish that some $\Z_{p^s}$-linear GH codes of length $p^t$ are
equivalent, once $t$ is fixed. This allows us to improve the known upper bounds for the number of such nonequivalent codes.
Moreover, up to $t=10$,  this new upper bound coincides with a known lower bound (based on the rank and dimension of the kernel). 
\end{abstract}



\section{Introduction}\label{intro}
Let $\Z_{p^s}$ be the ring of integers modulo $p^s$ with $p$ prime and $s\geq1$. The set of
$n$-tuples over $\Z_{p^s}$ is denoted by $\Z_{p^s}^n$. In this paper,
the elements of $\Z^n_{p^s}$ will also be called vectors. 
The order of a vector $\mathbf u$ over $\Z_{p^s}$, denoted by $o(\mathbf{u})$, is the smallest positive integer $m$ such that $m \mathbf{u} =\zero$.

A code over $\Z_p$ of length $n$ is a nonempty subset of $\Z_p^n$,
and it is linear if it is a subspace of $\Z_{p}^n$. Similarly, a nonempty
subset of $\Z_{p^s}^n$ is a $\Z_{p^s}$-additive if it is a subgroup of $\Z_{p^s}^n$.
Note that, when $p=2$ and $s=1$, a $\Z_{p^s}$-additive code is a binary linear code and, when $p=2$ and $s=2$ ,
it is a quaternary linear code or a linear code over $\Z_4$.


In \cite{Sole}, a Gray map  from $\Z_4$ to $\Z_2^2$ is defined as
$\phi(0)=(0,0)$, $\phi(1)=(0,1)$, $\phi(2)=(1,1)$ and $\phi(3)=(1,0)$. There exist different generalizations of this Gray map, which go from $\Z_{2^s}$ to
$\Z_2^{2^{s-1}}$ \cite{Carlet,Codes2k,dougherty,Nechaev,Krotov:2007}.
The one given in \cite{Nechaev} can be defined in terms of the elements of a Hadamard code \cite{Krotov:2007}, and Carlet's Gray map \cite{Carlet} is a particular case of the one given in \cite{Krotov:2007} 
satisfying $\sum \lambda_i \phi(2^i) =\phi(\sum \lambda_i 2^i)$ \cite{KernelZ2s}. 
In this paper, we focus on a generalization of Carlet's Gray map, from $\Z_{p^s}$ to $\Z_p^{p^{s-1}}$, which is also a particular case of the one given in \cite{ShiKrotov2019}. Specifically,
\begin{equation}\label{genGraymap}
\phi_s(u)=(u_{s-1},\dots,u_{s-1})+(u_0,\dots,u_{s-2})Y_{s-1},
\end{equation}
where $u\in\Z_{p^s}$, $[u_0,u_1, \ldots, u_{s-1}]_p$ is the p-ary expansion of $u$, that is $u=\sum_{i=0}^{s-1}p^{i}u_i$ ($u_i \in \Z_p$), and $Y_{s-1}$ is a matrix of size $(s-1)\times p^{s-1}$ whose columns are the elements of $\Z_p^{s-1}$. 
Then, we define $\Phi_s:\Z_{p^s}^ n\rightarrow\Z_p^{np^{s-1}}$ as the component-wise Gray map $\phi_s$.

Let $\C$ be a $\Z_{p^s}$-additive code of length $n$. We say that its image
$C=\Phi(\C)$ is a $\Z_{p^s}$-linear code of length $p^{s-1}n$.
Since $\C$ is a subgroup of
$\Z_{p^s}^n$, it is isomorphic to an abelian structure
$\Z_{p^s}^{t_1}\times\Z_{p^{s-1}}^{t_2}\times
\dots\times\Z_{p^2}^{t_{s-1}}\times\Z_p^{t_s}$, and we say that $\C$, or equivalently
$C=\Phi(\C)$, is of type $(n;t_1,\dots,t_{s})$.
Note that $|\C|=p^{st_1}p^{(s-1)t_2}\cdots p^{t_s}$.
Unlike linear codes over finite fields,
linear codes over rings do not have a basis, but there
exists a generator matrix for these codes having minimum number of rows, that is, $t_1+\cdots+t_s$ rows. 
A $p$-linear combination of the elements of $\mathcal{B}=\{\mathbf{b}_1,\dots,\mathbf{b}_r\}\subseteq \Z_{p^s}^n$ is $\sum_{i=1}^{r}\lambda_i\mathbf{b}_i$, for $\lambda_i\in\Z_p$. We say that $\mathcal{B}$ is a $p$-basis of $\C$ if the elements in $\mathcal{B}$ are $p$-linearly independent and any $\mathbf{c}\in\C$ is a $p$-linear combination of the elements of  $\cal{B}$.

Let $\cS_n$ be the symmetric group of permutations on the set $\{1,\dots,n\}$.
Two codes $C_1$ and $C_2$ over $\Z_p$, are said to be equivalent if there is a vector $\textbf{a}\in \Z_p^n$ and a
permutation of coordinates $\pi\in \cS_n$ such that $C_2=\{ \textbf{a}+\pi(\textbf{c}) : \textbf{c} \in C_1 \}$.
Two $\Z_{p^s}$-additive codes, $\C_1$ and $\C_2$, are said to be permutation equivalent if they differ
only by a permutation of coordinates, that is, if there is a permutation of coordinates $\pi\in \cS_n$
such that $\C_2=\{ \pi(\textbf{c}) : \textbf{c} \in \C_1 \}$.

Two structural properties of codes over $\Z_p$ are the rank and
dimension of the kernel. The rank of a code $C$ over $\Z_p$ is simply the
dimension of the linear span, $\langle C \rangle$,  of $C$.
The kernel of a code $C$ over $\Z_p$ is defined as
$\mathrm{K}(C)=\{\textbf{x}\in \Z_p^n : \textbf{x}+C=C \}$ \cite{BGH83,pKernel}. If the all-zero vector belongs to $C$,
then $\mathrm{K}(C)$ is a linear subcode of $C$.
Note also that if $C$ is linear, then $K(C)=C=\langle C \rangle$.
We denote the rank of $C$ as $\rank(C)$ and the dimension of the kernel as $\kernel(C)$.
These parameters can be used to distinguish between non-equivalent codes, since equivalent ones have the same rank and dimension of the kernel.

A generalized Hadamard $(GH)$ matrix $H(p,\lambda) = (h_{i j})$ of order $n = p\lambda$ over $\Z_p$ is a $p\lambda \times p\lambda$ matrix with entries from $\Z_p$ with the property that for every $i, j$, $1 \leq i < j \leq p\lambda,$ each of the multisets $\{h_{is}- h_{js} : 1 \leq s \leq p\lambda\}$ contains every element of $\Z_p$ exactly $\lambda$ times \cite{jungnickel1979}. 
An ordinary Hadamard matrix of order $4\mu$ corresponds to  $GH$ matrix $H(2,\lambda)$ over $\Z_2$, where $\lambda = 2\mu$ \cite{Key}. 
Two $GH$ matrices $H_1$ and $H_2$ of order $n$ are said to be equivalent if one can be obtained from the other by a permutation of the rows and columns and adding the same element of $\Z_p$ to all the coordinates in a row or in a column. We can always change the first row and column of a $GH$ matrix into zeros and we obtain an equivalent $GH$ matrix which is called normalized. From a normalized Hadamard matrix $H$, we denote by $F_H$ the code over $\Z_p$ consisting of the rows of $H$, and $C_H$ the one defined as $C_H = \bigcup_{\alpha \in\Z_p} (F_H + \alpha \textbf{1})$,
where $F_H + \alpha \textbf{1} = \{\textbf{h} + \alpha \textbf{1} : \textbf{h} \in F_H\}$ and $\textbf{1}$ denotes the
all-one vector. The code $C_H$ over $\Z_p$ is called generalized
Hadamard $(GH)$ code \cite{dougherty2015ranks}. Note that $C_H$ is generally a nonlinear code over $\Z_p$.

Let $\C$ be a $\Z_{p^s}$-additive code such that $\Phi(\C)$ is a GH code. Then, we say that $\C$ is a $\Z_{p^s}$-additive GH code and $\Phi(\C)$ is a $\Z_{p^s}$-linear GH code.
Note that a GH code over $\Z_p$ of length $N$ has $pN$ codewords and minimum distance $\frac{N(p-1)}{p}$.

The $\Z_4$-linear Hadamard codes of length $2^t$ can be classified by using either the rank or the dimension of the kernel \cite{Kro:2001:Z4_Had_Perf,PheRifVil:2006}. There are exactly $\lfloor
\frac{t-1}{2}\rfloor$ such codes for all $t\geq 2$. 
Later, in \cite{KernelZ2s}, an iterative construction for $\Z_{2^s}$-linear Hadamard codes was described, and the linearity and kernel of these codes were established. It was proved that the dimension of the kernel only provides a complete classification for some values of $t$ and $s$. A partial classification by using the kernel was obtained, and the exact amount of non-equivalent such codes was given up to $t=11$ for any $s\geq 2$. Authors also give a lower bound based on the rank and dimension of the kernel. Later, in paper \cite{EquivZ2s}, for $s\geq 2$, it was established that some $\Z_{2^s}$-linear Hadamard codes of length $2^t$ are equivalent, once $t$ is fixed, and this fact improved the known upper bounds for the number of such nonequivalent codes. Moreover, authors show that, up to $t = 11$, this new upper bound coincides with the known lower bound (based on the rank and dimension of the kernel). Finally, for $s \in \{2,3\}$, they establish the full classification of the $\Z_{2^s}$-linear Hadamard codes of length $2^t$  by giving the exact number of such codes.
For $s\geq 2$, and $p\geq 3$ prime, the dimension of the kernel for $\Z_{p^s}$-linear GH codes of length $p^t$  is established in \cite{HadamardZps}, and it is proved that this invariant only provides a complete classification for certain values of $t$ and $s$. Lower and upper bounds are also established for the number of nonequivalent $\Z_{p^s}$-linear GH codes of length $p^t$, when both $t$ and $s$ are fixed, and when just $t$ is fixed; denoted by $\mathcal{A}_{t,s,p}$ and $\mathcal{A}_{t,p}$, respectively. 
From \cite{HadamardZps}, we can check that there are nonlinear codes having the same rank and dimension of the kernel for different values of $s$, once the length $p^t$ is fixed, for all $4\leq t\leq 11$.

In this paper, we show that there are generalized Hadamard codes of length $p^t$, with $p\geq 3$ and $t\geq 1$, that can be seen up to permutations as $\Z_{p^s}$-linear codes for different values of $s$. Moreover, for $t\leq 10$, we see that the codes that are permutation equivalent are, in fact, those having the same pair of invariants, rank and dimension of the kernel. For example, in Table \ref{table:Types}, for $t=7$, the colored codes having same rank and dimension of the kernel are permutation equivalent. These equivalence results  allow us to obtain a more accurate classification of the $\Z_{p^s}$-linear GH codes, than the one given in \cite{HadamardZps}. 

The paper is organized as follows. In Section \ref{Sec:construction}, we recall the recursive construction of the $\Z_{p^s}$-linear GH codes, the known partial classification, and some bounds on the number of nonequivalent such codes, presented in \cite{HadamardZps}. In Section \ref{sec:relations}, we prove some equivalence relations among the $\Z_{p^s}$-linear GH codes of the same length $p^t$. 
Later, in Section \ref{sec:Classification}, we improve the classification given in \cite{HadamardZps} by refining the known bounds. 

\section{Preliminaries and Partial classification}
\label{Sec:construction}

In this section, we provide some results presented in \cite{HadamardZps} and related to 
the recursive construction, partial classification and
bounds on the number of nonequivalent $\Z_{p^s}$-linear GH codes of length $p^t$ with $t\geq 3$ and $p\geq 3$ prime.

Let $T_i=\lbrace j\cdot p^{i-1}\, :\, j\in\lbrace0,1,\dots,p^{s-i+1}-1\rbrace \rbrace$ for all $i \in \{1,\ldots,s \}$.
Note that $T_1=\lbrace0,\dots,p^{s}-1\rbrace$. Let $t_1$, $t_2$,\dots,$t_s$ be nonnegative integers with $t_1\geq1$. Consider the matrix $A_p^{t_1,\dots,t_s}$ whose columns are exactly all the vectors of the form $\mathbf{z}^T$, $\mathbf{z}\in\lbrace1\rbrace\times T_1^{t_1-1}\times T_{2}^{t_2}\times\cdots\times T_s^{t_s}$. We write $A^{t_1,\dots,t_s}$ instead of $A_p^{t_1,\dots,t_s}$ when the value of $p$ is clear by the context.
Let  $\mathbf{0}, \mathbf{1},\mathbf{2},\ldots, \mathbf{p^{s}-1}$ be the vectors having the same element $0, 1, 2, \ldots, p^s-1$ from $\Z_{p^s}$ in all its coordinates, respectively. 

Any matrix $A^{t_1,\dots,t_s}$ can be obtained by applying the
following recursive construction. We start with $A^{1,0,\dots,0}=(1)$. Then, if
we have a matrix $A=A^{t_1,\dots,t_s}$, for any $i\in \{1,\ldots,s\}$, we may construct the matrix
\begin{equation}\label{eq:recGenMatrix}
A_i=
\left(\begin{array}{cccc}
A & A &\cdots & A \\
0\cdot \mathbf{p^{i-1}}  & 1\cdot \mathbf{p^{i-1}} & \cdots & (p^{s-i+1}-1)\cdot \mathbf{p^{i-1}}  \\
\end{array}\right).
\end{equation}
Finally, permuting the rows of $A_i$, we obtain a matrix $A^{t'_1,\ldots,t'_s}$, where $t'_j=t_j$ for $j\not=i$ and $t'_i=t_i+1$. Note that any permutation of columns of $A_i$ gives also a matrix $A^{t_1',\dots,t_s'}$.
Along this paper, we consider that the matrices $A^{t_1,\ldots,t_s}$ are constructed recursively starting from $A^{1,0,\ldots,0}$ in the following way. First, we add $t_1-1$ rows of order $p^s$, up to obtain $A^{t_1,0,\ldots,0}$; then $t_2$ rows of order $p^{s-1}$ up to generate $A^{t_1,t_2,\ldots,0}$; and so on, until we add $t_s$ rows of order $p$ to achieve $A^{t_1,\ldots,t_s}$. See \cite{HadamardZps} for examples.

Let $\mathcal{H}^{t_1,\dots,t_s}$ be the $\Z_{p^s}$-additive code generated by the matrix $A^{t_1,\dots,t_s}$, where $t_1,\dots,t_s$ are nonnegative integers with $t_1\geq1$. Let $n=p^{t-s+1}$, where $t=\left(\sum_{i=1}^{s}(s-i+1)\cdot t_i\right)-1$. The code $\mathcal{H}^{t_1,\dots,t_s}$ has length $n$, and the corresponding $\Z_{p^s}$-linear code $H^{t_1,\dots,t_s}=\Phi(\mathcal{H}^{t_1,\dots,t_s})$ is a GH code of length $p^t$ \cite{HadamardZps}. In order to classify the $\Z_{p^s}$-linear GH codes of length $p^t$, we can focus on $t\geq 5$ and $2\leq s \leq t-2$ for $p=2$ (\cite{KernelZ2s}), and on $t\geq 4$ and $2\leq s \leq t-2$ for $p\geq 3$ prime (\cite{HadamardZps}). Moreover, for any $t\geq 5$ and $2\leq s \leq t-2$, there are two $\Z_{2^s}$-linear Hadamard codes of length $2^t$ which are linear, say  $H^{1,0,\ldots,0,t+1-s}$ and $H^{1,0\dots,0,1,t-1-s}$; and for any $t\geq 4$ and $2\leq s \leq t-2$, there is a unique $\Z_{p^s}$-linear GH code of length $p^t$ which is linear, say $H^{1,0,\ldots,0,t+1-s}$. 

Tables \ref{table:Types} and \ref{table:TypesB2}, for $4\leq t\leq 10$ and $2\leq s\leq t-2$, show all possible values
$(t_1,\dots,t_s)$  for which there exists a nonlinear $\Z_{p^s}$-linear Hadamard code $H^{t_1,\ldots,t_s}$ of length $p^t$.
For each one of them, the values $(r,k)$ are shown, where $r$ is the rank (computed by using the computer algebra system Magma \cite{Magma,Qarypackage}) and $k$ is the dimension of the kernel (\cite{HadamardZps,KernelZ2s,Kro:2001:Z4_Had_Perf}).  Note that if two codes have different values $(r,k)$, then they are not equivalent. Therefore, from the values of the dimension of the kernel given in these tables, it is easy to see that this invariant does not classify. Also, considering only the rank, it is not possible to fully classify these codes either.

\begin{table}[ht!]
\footnotesize
\centering
\begin{tabular}{|c||cc|cc|cc|cc|}
\cline{1-9}
& \multicolumn{2}{c|}{$t=4$}& \multicolumn{2}{c|}{$t=5$}& \multicolumn{2}{c|}{$t=6$}& \multicolumn{2}{c|}{$t=7$}\\
\cline{2-9}
& $(t_1,\ldots,t_s)$ & $(r,k)$ & $(t_1,\ldots,t_s)$ & $(r,k)$ & $(t_1,\ldots,t_s)$ & $(r,k)$ & $(t_1,\ldots,t_s)$ & $(r,k)$\\[0.2em]
\hline
\multirow{2}{*}{$\Z_9$} &$(2,1)$  & (6,3)  &$(2,2)$  &(7,4)    & $(3,1)$  & (12,4)  & $(3,2)$  & (13,5)  \\
                        &&       & $(3,0)$  & (11,3)             &$(2,3)$  &(8,5)                & $(4,0)$  & (21,4)\\
                        &&    &&    &&  &$\textcolor{blue}{(2,4)}$ & \textcolor{blue}{(9,6)}  \\[0.2em]
\hline

\multirow{5}{*}{$\Z_{3^3}$} &$(1,1,0)$ & (6,3)    & $(2,0,0)$ & (13,2) & $(1,2,0)$ & (12,4)  & $(1,2,1)$ & (13,5) \\
  &&                    &$(1,1,1)$           &(7,4)       & $(2,0,1)$ & (14,3)      & $(2,0,2)$ & (15,4)  \\
  &&                    &           &       & $(1,1,2)$          &(8,5)             &  $(2,1,0)$ & (25,3)  \\
   &&                   &           &       &           &              &$\textcolor{blue}{(1,1,3)}$          &\textcolor{blue}{(9,6)}    \\[0.2em]  
\hline

\multirow{4}{*}{$\Z_{3^4}$}  && &$(1,0,1,0)$ & (7,4)       & $(1,1,0,0)$ & (14,3)  & $(1,0,2,0)$ & (13,5)   \\
    &&                       &             &             & $(1,0,1,1)$ &(8,5)    & $(1,1,0,1)$ & (15,4)  \\
    &&                       &             &       &             &        & $(2,0,0,0)$ & (14,2) \\
    &&                       &             &       &             &        &$\textcolor{blue}{(1,0,1,2)}$ &\textcolor{blue}{(9,6)}        \\[0.2em]

\hline
\multirow{3}{*}{$\Z_{3^5}$} && &             &       &$(1,0,0,1,0)$ &(8,5)        & $(1,0,1,0,0)$ & (15,4) \\
    &&                       &             &       &             &        & $\textcolor{blue}{(1,0,0,1,1)}$ &\textcolor{blue}{(9,6)}        \\[0.2em]
    
\hline
\multirow{1}{*}{$\Z_{3^6}$} &&  &             &       &             &        &$\textcolor{blue}{(1,0,0,0,1,0)}$               &\textcolor{blue}{(9,6)}      \\[0.2em]
\hline

\end{tabular}
\caption{Rank and kernel for all nonlinear $\Z_{3^s}$-linear GH codes of 
length $3^t$.}
\label{table:Types}
\end{table}

Let $X_{t,s,p}$ be the number of nonnegative integer solutions $(t_1,\ldots,t_s)\in \N^s$ of the equation $t=\left(\sum_{i=1}^{s}(s-i+1)\cdot t_i\right)-1$ with $t_1\geq 1$. This gives the number of $(t_1,\ldots,t_s)$ so that $H^{t_1,\ldots,t_s}$ is a $\Z_{p^s}$-linear GH codes of length $p^t$.
Let $\cA_{t,s,p}$ be the number of nonequivalent $\Z_{p^s}$-linear GH codes of length $p^t$ and a fixed $s\geq 2$. Then, for any $t \geq 5$ and $2 \leq s \leq t-1$, we have that $\cA_{t,s,2} \leq X_{t,s,2}-1$, since there are exactly two codes which are linear \cite{KernelZ2s}. For $p\geq 3$ prime, $t \geq 4$ and $2 \leq s \leq t-1$, we have that $\cA_{t,s,p} \leq X_{t,s,p}$, since there are exactly one code which is linear.
Moreover, for  $t \leq 10$ if $p\geq 3$ prime and for $t \leq 11$ if $p=2$ , this bound is tight. It is still an open problem to know whether this bound is tight for $t\geq 11$ if $p\geq 3$ prime and $t\geq 12$ if $p=2$.

A partial classification for the $\Z_{p^s}$-linear GH codes of length $p^t$ is given in \cite{HadamardZps} for $p\geq 3$ prime, and for $p=2$, it is given in \cite{KernelZ2s}.
Specifically, lower and upper bounds on the number of nonequivalent such codes, once only $t$ is fixed, are established. We first consider the case $p=2$. 

\begin{theorem}\label{theo:CardinalFirstBoundY} \cite{KernelZ2s}
Let $\cA_{t,2}$ be the number of nonequivalent $\Z_{2^s}$-linear Hadamard codes of length $2^t$ with $t\geq 3$. Then,
\begin{equation}\label{eq:CardinalFirstBound}
\cA_{t,2} \leq 1+ \sum_{s=2}^{t-2} (X_{t,s,2}-2)
\end{equation}
and
\begin{equation}\label{eq:CardinalFirstBoundY1}
\cA_{t,2} \leq 1+\sum_{s=2}^{t-2} (\cA_{t,s,2}-1).
\end{equation}
\end{theorem}

For $p=2$, the upper bound of the value of $\cA_{t,2}$ is improved in \cite{EquivZ2s} where it was proved that for a fixed $t$ there are $\Z_{2^s}$-linear Hadamard codes of length $2^t$ that are equivalent. Let $\tilde{X}_{t,s,2}=|\{ (t_1,\ldots,t_s)\in \N^s :  t+1=\sum_{i=1}^{s}(s-i+1) t_i, \ t_1\geq 2 \}|$ for $s\in\{3,\dots,\lfloor (t+1)/2\rfloor\}$ and $\tilde{X}_{t,2,2}=|\{ (t_1,t_2)\in \N^2 :  t+1=2t_1+t_2, \ t_1\geq 3 \}|$. 

\begin{theorem}[\cite{EquivZ2s}] \label{teo:numNonEquiv100}
For $t \geq 3$,
\begin{equation}\label{eq:BoundAtNewX}
\cA_{t,2} \leq1+ \sum_{s=2}^{\lfloor\frac{t+1}{2}\rfloor} \tilde{X}_{t,s,2}
\end{equation}
and
\begin{equation}\label{eq:BoundAtNew}
\cA_{t,2} \leq1+ \sum_{s=2}^{\lfloor\frac{t+1}{2}\rfloor} ({\cA}_{t,s,2}-1).
\end{equation}
Moreover, for $3 \leq t \leq 11$, first bound is tight.
\end{theorem}  
It is also proved that, for $5\leq t \leq 11$, the improved upper bound is the exact value of $\cA_{t,2}$. Now we consider the case $p\geq 3$.

\begin{theorem}\label{theo:CardinalFirstBound} \cite{HadamardZps}
Let $\cA_{t,p}$ be the number of nonequivalent $\Z_{p^s}$-linear GH codes of length $p^t$ with $t\geq 3$ and $p\geq 3$ prime. Then,
\begin{equation}\label{eq:CardinalFirstBound-p}
\cA_{t,p} \leq 1+ \sum_{s=2}^{t-1} (X_{t,s,p}-1)
\end{equation}
and
\begin{equation}\label{eq:CardinalFirstBoundX-p}
\cA_{t,p} \leq 1+\sum_{s=2}^{t-1} (\cA_{t,s,p}-1).
\end{equation}
\end{theorem}


In this paper, in order to improve this partial classification,
we analyze the equivalence relations among the $\Z_{p^s}$-linear GH codes with the same length $p^t$ and
different values of $s$. We prove that some of them are indeed permutation equivalent.
For $5\leq t\leq 11$ if $p=2$ and  for $4\leq t\leq 10$ if $p\geq 3$ prime, the ones that are permutation equivalent coincide with the ones that have the same invariants, rank and dimension of the kernel,
that is, the same pair $(r,k)$ in the Tables. 
Finally, by using this equivalence relations,  we improve the upper bounds on the number  $\cA_{t,p}$ of nonequivalent
$\Z_{p^s}$-linear GH codes of length $p^t$ with $p\geq 3$ prime given by Theorem \ref{theo:CardinalFirstBound}.
This allows us to determine the exact value of $\cA_{t,p}$ for $5\leq t\leq 11$ if $p=2$ and for $4\leq t \leq 10$ if $p\geq 3$ prime, since one of the new upper bounds coincides with the lower bound $(r,k)$ for these cases.

\begin{table}[H]
\footnotesize
\centering
\begin{tabular}{|c||cc|cc|cc|}
\cline{1-7}
& \multicolumn{2}{c|}{$t=8$}& \multicolumn{2}{c|}{$t=9$}& \multicolumn{2}{c|}{$t=10$}\\
\cline{2-7}
& $(t_1,\ldots,t_s)$ & $(r,k)$ & $(t_1,\ldots,t_s)$ & $(r,k)$ & $(t_1,\ldots,t_s)$ & $(r,k)$ \\[0.2em]
\hline
\multirow{4}{*}{$\Z_9$} & $(3,3)$  & (14,6) & $(3,4)$ & (15,7)  & $(3,5)$ & (16,8)   \\
                        & $(4,1)$  & (22,5) & $(4,2)$ & (23,6)  & $(4,3)$ & (24,7)   \\
                        &$(2,5)$   & (10,7)    & $(5,0)$ & (36,5)  & $(5,1)$ & (37,6)  \\[0.2em]
                        &          &        &$(2,6)$  & (11,8)     & $(2,7)$ & (12,9)   \\[0.2em]
\hline
\multirow{10}{*}{$\Z_{3^3}$}& $(1,2,2)$ & (14,6) & $(1,2,3)$ & (15,7) & $(1,2,4)$ & (16,8)   \\
                    & $(1,3,0)$ & (22,5)    & $(1,3,1)$ & (23,6) & $(1,3,2)$ & (24,7)  \\
                     & $(2,0,3)$ & (16,5)   & $(2,0,4)$ & (17,6) & $(1,4,0)$ & (37,6)  \\
                     & $(2,1,1)$ & (26,4)   & $(2,1,2)$ & (27,5) & $(2,0,5)$ & (18,7)  \\
                     &$(3,0,0)$ & (48,3)   & $(2,2,0)$ & (43,4) & $(2,1,3)$ & (28,6)  \\
                     &$(1,1,4)$ & (10,7)   & $(3,0,1)$ & (49,4) & $(2,2,1)$ & (44,5)  \\
                     &&   & $(1,1,5)$ & (11,8)    & $(3,0,2)$ & (50,5)  \\
                     &&   &           &        & $(3,1,0)$ & (82,4)  \\
                     &&   &           &        & $(1,1,6)$ & (12,9)     \\[0.2em]
\hline
\multirow{13}{*}{$\Z_{3^4}$}& $(1,0,2,1)$ & (14,6)  & $(1,0,2,2)$ & (15,7) & $(1,0,2,3)$ & (16,8)   \\
                          & $(1,1,0,2)$ & (16,5) & $(1,0,3,0)$ & (23,6) & $(1,0,3,1)$ & (24,7)   \\
                          & $(1,1,1,0)$ & (26,4) & $(1,2,0,0)$ & (49,4) & $(1,1,0,4)$ & (18,7)  \\
                          & $(2,0,0,1)$ & (35,3) & $(1,1,0,3)$ & (17,6) & $(1,1,1,2)$ & (28,6)  \\
                          &$(1,0,1,3)$&(10,7) & $(1,1,1,1)$ & (27,5) & $(1,1,2,0)$ & (44,5)  \\
                          && & $(2,0,0,2)$ & (36,4) & $(1,2,0,1)$ & (50,5)  \\
                          && & $(2,0,1,0)$ & (64,3) & $(2,0,0,3)$ & (37,5)  \\
                          && & $(1,0,1,4)$  & (11,8)       & $(2,0,1,1)$ & (65,4)  \\
                          && &             &        & $(2,1,0,0)$ & (121,3)  \\
                          && &             &        & $(1,0,1,5)$  & (12,9)        \\[0.2em]
\hline
\multirow{11}{*}{$\Z_{3^5}$}& $(1,0,0,2,0)$ & (14,6)  & $(1,0,0,2,1)$ & (15,7) & $(1,0,0,2,2)$ & (16,8)  \\
                          & $(1,0,1,0,1)$ & (16,5) & $(1,0,1,0,2)$ & (17,6) & $(1,0,0,3,0)$ & (24,7) \\
                          & $(1,1,0,0,0)$ & (35,3) & $(1,0,1,1,0)$ & (27,5) & $(1,0,1,0,3)$ & (18,7)  \\
                          &$(1,0,0,1,2)$ &(10,7) & $(1,1,0,0,1)$ & (36,4) & $(1,0,1,1,1)$ & (28,6)  \\
                          && & $(2,0,0,0,0)$ & (96,2) & $(1,0,2,0,0)$ & (50,5)  \\
                          && & $(1,0,0,1,3)$ & (11,8)                  & $(1,1,0,0,2)$ & (37,5)  \\
                          && &               &        & $(1,1,0,1,0)$ & (65,4)  \\
                          && &               &        & $(2,0,0,0,1)$ & (97,3)  \\
                          && &               &        & $(1,0,0,1,4)$ & (12,9)  \\[0.2em]
\hline
\multirow{9}{*}{$\Z_{3^6}$} & $(1,0,0,1,0,0)$ & (16,5) & $(1,0,0,0,2,0)$ & (15,7)  & $(1,0,0,0,2,1)$ & (16,8)  \\
                          &$(1,0,0,0,1,1)$ &(10,7) & $(1,0,0,1,0,1)$ & (17,6)  & $(1,0,0,1,0,2)$ & (18,7) \\
                          && & $(1,0,1,0,0,0)$ & (36,4)  & $(1,0,0,1,1,0)$ & (28,6)  \\
                          && & $(1,0,0,0,1,2)$ & (11,8)     & $(1,0,1,0,0,1)$ & (37,5)  \\
                          && &                 &         & $(1,1,0,0,0,0)$ & (97,3) \\
                          && &                 &         & $(1,0,0,0,1,3)$ & (12,9) \\[0.2em]
\hline
\multirow{5}{*}{$\Z_{3^7}$}&$(1,0,0,0,0,1,0)$ &(10,7)& $(1,0,0,0,1,0,0)$ & (17,6)  & $(1,0,0,0,0,2,0)$ & (16,8)  \\
                          && & $(1,0,0,0,0,1,1)$ & (11,8)         & $(1,0,0,0,1,0,1)$ & (18,7) \\
                          && &                   &         & $(1,0,0,1,0,0,0)$ & (37,5) \\
                          && &                   &         & $(1,0,0,0,0,1,2)$ &(12,9)         \\[0.2em]
\hline
\multirow{3}{*}{$\Z_{3^8}$}&& & $(1,0,0,0,0,0,1,0)$ &(11,8)         & $(1,0,0,0,0,1,0,0)$ & (18,7) \\
                           && &                 &         &  $(1,0,0,0,0,0,1,1)$ &(12,9)         \\ [0.2em]
                        
\hline
\multirow{1}{*}{$\Z_{3^9}$}&& & & &$(1,0,0,0,0,0,0,1,0)$ &(12,9)\\[0.2em]
\hline
\end{tabular}
\caption{Rank and kernel for all nonlinear $\Z_{3^s}$-linear GH codes of 
length $3^t$.}
\label{table:TypesB2}
\end{table}

\section{Equivalent $\Z_{p^s}$-linear GH codes}
\label{sec:relations}

In this section, we give some properties of the generalized Gray map $\phi_s$.
We also prove that, for $p\geq 3$ prime, some of the $\Z_{p^s}$-linear GH codes of the same length $p^t$,
having different values of $s$ are permutation equivalent.
Moreover, we see that they coincide with the ones having the same rank and dimension of the kernel for $4\leq t\leq 10$.

\begin{lemma}\label{lema4}
Let $\lambda_i\in \Z_p$, $i\in\lbrace0,\dots,s-2\rbrace$. Then,
$$
\sum_{i=0}^{s-2}\lambda_i\phi_s(p^ i)=\phi_s(\sum_{i=0}^{s-2}\lambda_ip^i).
$$
\end{lemma}
\begin{proof}
Straightforward form the definition of $\phi_s$.
\end{proof}

\bigskip
Let $\gamma_s\in\mathcal{S}_{p^{s-1}}$ be the permutation defined as follows: for a coordinate $k=jp^{s-2}+i+1 \in \{1,2,\dots, p^{s-1}\}$, where $j\in \{0,\dots, p-1\}$ and $i\in \{0,\dots, p^{s-2}-1\}$, $\gamma_s$ moves coordinate $k$ to coordinate $j+ip+1$. 
So, we can write $\gamma_s$ as
\begin{equation*}\label{eq:gamma}
\tiny
\left(
  \begin{array}{ccccccccccccc}
    1 & 2&  \ldots & p^{s-2} & p^{s-2}+1 & p^{s-2}+2 &\ldots & p^{s-2}+p & \ldots & p^{s-1}-p+1 &p^{s-1}-p+2 &\ldots &p^{s-1} \\
    1 & p+1&  \ldots &(p^{s-2}-1)p+1 &2 & p+2 & \ldots & (p^{s-2}-1)p+2 &\ldots & p &p+p & \ldots  &p^{s-1}\\
  \end{array}
\right).
\end{equation*}

\begin{example}
For $p=3$ and $s=3$,
 \begin{equation*}
 \gamma_3=\left(
   \begin{array}{ccccccccccc}
     1 & 2& 3& & 4& 5& 6& &7& 8& 9 \\
     1 & 4& 7& & 2& 5& 8& &3& 6& 9 \\
   \end{array}
 \right)=(2,4)(3,7)(6,8) \in \mathcal{S}_9.
 \end{equation*}
 and for $p=3$ and $s=4$, 
 \begin{multline*}
 \footnotesize
 \arraycolsep=1.4pt\def\arraystretch{1}
 \gamma_4=\left(
   \begin{array}{ccccccccccccccccccccccccccccc}
     1 & 2& 3&  4&  5&  6&  7&  8&  9& & 10& 11& 12& 13& 14& 15& 16& 17& 18 &  &19& 20& 21& 22& 23& 24& 25& 26& 27 \\
     1 & 4& 7& 10& 13& 16& 19& 22& 25& &  2&  5&  8& 11& 14& 17& 20& 23& 26 &  & 3&  6&  9& 12& 15& 18& 21& 24& 27 \\
 \end{array}
 \right)\\
 =(2,4,10)(3,7,19)(5,13,11)(6,16,20)(8,22,12)(9,25,21)(15,17,23)(18,26,24)\in \mathcal{S}_{27}.
 \end{multline*}
\end{example}

\begin{lemma}\label{lem:gamma}
Let $s\geq 2$, $u=(\zero, \one, \dots, \pmo)\in \Z_p^{p^{s-1}}$ and $v=(0,1,\dots,p-1)\in\Z_p^p$. Then,
\begin{enumerate}[label=(\roman*)]
\item $\gamma_s(u)=(v, \stackrel{p^{s-2}}{\dots},v)$,
\item $\gamma_s(v, \stackrel{p^{s-2}}{\dots},v)=u$. 
\end{enumerate}
\end{lemma}
\begin{proof}
Straightforward form the definition of $\gamma_s$.
\end{proof}

Then, we can define the map $\tau_s:\Z_{p^s}\rightarrow \Z_{p^{s-1}}^p$ as
\begin{equation}\label{eq:tau}
\tau_s(u)=\Phi_{s-1}^{-1}(\gamma_s^{-1}(\phi_s(u))),
\end{equation}
where $u\in\Z_{p^s}$.

\begin{example}
For $p=3$ and $s=3$, we have
\footnotesize
\begin{equation}\label{eq:equivbase}
\left.\begin{array}{cccccccl}
\phi_3(0) & = & (0,0,0,0,0,0,0,0,0) & = &\gamma_3(0,0,0,0,0,0,0,0,0) = \gamma_3(\Phi_2(0,0,0))&\\
\phi_3(1) & = & (0,1,2,0,1,2,0,1,2) & = &\gamma_3(0,0,0,1,1,1,2,2,2) = \gamma_3(\Phi_2(0,3,6))&\\
\phi_3(2) & = & (0,2,1,0,2,1,0,2,1) & = &\gamma_3(0,0,0,2,2,2,1,1,1) = \gamma_3(\Phi_2(0,6,3))&\\
\phi_3(3) & = & (0,0,0,1,1,1,2,2,2) & = &\gamma_3(0,1,2,0,1,2,0,1,2) = \gamma_3(\Phi_2(1,1,1))&\\
\phi_3(4) & = & (0,1,2,1,2,0,2,0,1) & = &\gamma_3(0,1,2,1,2,0,2,0,1) = \gamma_3(\Phi_2(1,4,7))&\\
\phi_3(5) & = & (0,2,1,1,0,2,2,1,0) & = &\gamma_3(0,1,2,2,0,1,1,2,0) = \gamma_3(\Phi_2(1,7,4))&\\
\phi_3(6) & = & (0,0,0,2,2,2,1,1,1) & = &\gamma_3(0,2,1,0,2,1,0,2,1) = \gamma_3(\Phi_2(2,2,2))&\\
\phi_3(7) & = & (0,1,2,2,0,1,1,2,0) & = &\gamma_3(0,2,1,1,0,2,2,1,0) = \gamma_3(\Phi_2(2,5,8))&\\
\phi_3(8) & = & (0,2,1,2,1,0,1,0,2) & = &\gamma_3(0,2,1,2,1,0,1,0,2) = \gamma_3(\Phi_2(2,8,5))&\\
\phi_3(9) & = & (1,1,1,1,1,1,1,1,1) & = &\gamma_3(1,1,1,1,1,1,1,1,1) = \gamma_3(\Phi_2(3,3,3))&\\
\phi_3(10) & = & (1,2,0,1,2,0,1,2,0) & = &\gamma_3(1,1,1,2,2,2,0,0,0) = \gamma_3(\Phi_2(3,6,0))&\\
\phi_3(11) & = & (1,0,2,1,0,2,1,0,2) & = &\gamma_3(1,1,1,0,0,0,2,2,2) = \gamma_3(\Phi_2(3,0,6))&\\
\phi_3(12) & = & (1,1,1,2,2,2,0,0,0) & = &\gamma_3(1,2,0,1,2,0,1,2,0) = \gamma_3(\Phi_2(4,4,4))&\\
\phi_3(13) & = & (1,2,0,2,0,1,0,1,2) & = &\gamma_3(1,2,0,2,0,1,0,1,2) = \gamma_3(\Phi_2(4,7,1))&\\
\phi_3(14) & = & (1,0,2,2,1,0,0,2,1) & = &\gamma_3(1,2,0,0,1,2,2,0,1) = \gamma_3(\Phi_2(4,1,7))&\\
\phi_3(15) & = & (1,1,1,0,0,0,2,2,2) & = &\gamma_3(1,0,2,1,0,2,1,0,2) = \gamma_3(\Phi_2(5,5,5))&\\
\phi_3(16) & = & (1,2,0,0,1,2,2,0,1) & = &\gamma_3(1,0,2,2,1,0,0,2,1) = \gamma_3(\Phi_2(5,8,2))&\\
\phi_3(17) & = & (1,0,2,0,2,1,2,1,0) & = &\gamma_3(1,0,2,0,2,1,2,1,0) = \gamma_3(\Phi_2(5,2,8))&\\
\phi_3(18) & = & (2,2,2,2,2,2,2,2,2) & = &\gamma_3(2,2,2,2,2,2,2,2,2) = \gamma_3(\Phi_2(6,6,6))&\\
\phi_3(19) & = & (2,0,1,2,0,1,2,0,1) & = &\gamma_3(2,2,2,0,0,0,1,1,1) = \gamma_3(\Phi_2(6,0,3))&\\
\phi_3(20) & = & (2,1,0,2,1,0,2,1,0) & = &\gamma_3(2,2,2,1,1,1,0,0,0) = \gamma_3(\Phi_2(6,3,0))&\\
\phi_3(21) & = & (2,2,2,0,0,0,1,1,1) & = &\gamma_3(2,0,1,2,0,1,2,0,1) = \gamma_3(\Phi_2(7,7,7))&\\
\phi_3(22) & = & (2,0,1,0,1,2,1,2,0) & = &\gamma_3(2,0,1,0,1,2,1,2,0) = \gamma_3(\Phi_2(7,1,4))&\\
\phi_3(23) & = & (2,1,0,0,2,1,1,0,2) & = &\gamma_3(2,0,1,1,2,0,0,1,2) = \gamma_3(\Phi_2(7,4,1))&\\
\phi_3(24) & = & (2,2,2,1,1,1,0,0,0) & = &\gamma_3(2,1,0,2,1,0,2,1,0) = \gamma_3(\Phi_2(8,8,8))&\\
\phi_3(25) & = & (2,0,1,1,2,0,0,1,2) & = &\gamma_3(2,1,0,0,2,1,1,0,2) = \gamma_3(\Phi_2(8,2,5))&\\
\phi_3(26) & = & (2,1,0,1,0,2,0,2,1) & = &\gamma_3(2,1,0,1,0,2,0,2,1) = \gamma_3(\Phi_2(8,5,2))&\\
\end{array}\right.
\end{equation}

These equalities define the map $\tau_3: \Z_{27} \rightarrow \Z_9^3$ as
\begin{equation*}
\left.\begin{array}{cccccccccl}
     \tau_3(0)=(0,0,0) &&  \tau_3(9)  =(3,3,3)  && \tau_3(18)  =(6,6,6))\\
   \tau_3(1)=(0,3,6) &&   \tau_3(10)  =(3,6,0)    &&  \tau_3(19)  =(6,0,3)  \\ 
    \tau_3(2) =(0,6,3)  &&    \tau_3(11)  =(3,0,6)   &&    \tau_3(20)  =(6,3,0)  \\
     \tau_3(3)  =(1,1,1)  &&   \tau_3(12)  =(4,4,4)  &&\tau_3(21) =(7,7,7)  \\
      \tau_3(4)  =(1,4,7)  &&   \tau_3(13)  =(4,7,1)   && \tau_3(22)  =(7,1,4)   \\
       \tau_3(5)  =(1,7,4)  &&  \tau_3(14)  =(4,1,7)    &&  \tau_3(23)  =(7,4,1)  \\
        \tau_3(6)  =(2,2,2)  && \tau_3(15)  =(5,5,5)   &&  \tau_3(24)  =(8,8,8)  \\
         \tau_3(7)  =(2,5,8)  &&   \tau_3(16)  =(5,8,2)  &&  \tau_3(25)  =(8,2,5)  \\
          \tau_3(8)  =(2,8,5)  &&   \tau_3(17)  =(5,2,8)   && \tau_3(26) =(8,5,2) \\
\end{array}\right.
\end{equation*}
\end{example}

\begin{lemma}\cite{HadamardZps}\label{lem:phi}
Let $\lambda \in \Z_p$. Then, 
$\phi_s(\lambda p^{s-1})= (\lambda, \stackrel{p^{s-1}}{\dots},\lambda)$.
\end{lemma}
\begin{lemma} \label{lem:tau}
Let $s\geq 2$. Then,
\begin{enumerate}[label=(\roman*)]
\item $\tau_s(1)=(0,p^{s-2},\dots,(p-1)p^{s-2})$,
\item $\tau_s(p^iu)=p^{i-1}(u,\stackrel{p}{\dots},u)$ for all $i\in\{1,\dots,s-1\}$ and $u\in\{0,1,\dots,p^{s-1}-1\}$.
\end{enumerate}
\end{lemma}
\begin{proof}
First, let  $v=(0,1,\dots,p-1)\in\Z_p^p$. Then,
\begin{equation*}
    \begin{split}
        \tau_s(1) &=\Phi_{s-1}^{-1}(\gamma_s^{-1}(\phi_s(1)))\\
        &=\Phi_{s-1}^{-1}(\gamma_s^{-1}(v, \stackrel{p^{s-2}}{\dots},v))\\
        &=\Phi_{s-1}^{-1}(\mathbf{0},\mathbf{1},\dots,\mathbf{p-1}), \textrm {by Lemma \ref{lem:gamma}}.    
    \end{split}
\end{equation*}
Therefore by Lemma \ref{lem:phi},  $\tau_s(1)=(0,p^{s-2},\dots,(p-1)p^{s-2})$, and $(i)$ holds.

In order to prove $(ii)$, let $u\in\Z_{p^s}$ and $[u_0,\dots,u_{s-1}]_p$ be its p-ary expansion. The p-ary expansion of $p^{i}u$ is $[0,\stackrel{i}{\dots},0,u_0,\dots,u_{s-i-1}]_p$ and we have that $\phi_s(p^{i}u)=(u_{s-i-1},\dots,u_{s-i-1})+(0,\dots,0,u_0,\dots,u_{s-i-2})Y_{s-1}$.
Recall that the matrix $Y_{s-1}$ given in (\ref{genGraymap}), related to the definition of $\phi_s$,  is a matrix of size $(s-1)\times p^{s-1}$ whose columns are the elements of $\Z_p^{s-1}$. Without loss of generality, we consider that $Y_{s}$ is the matrix obtained recursively from $Y_{1}=(0~1~\cdots~p-1)$ and
\begin{equation}\label{eq:RecursiveYs}
Y_{s}=\left(\begin{array}{cccc}
Y_{s-1} & Y_{s-1} &\cdots & Y_{s-1}  \\
\zero & \one & \cdots & \mathbf{p-1}\\
\end{array}\right).
\end{equation}
By Lemma \ref{lem:gamma}, we can write

\begin{equation}
\gamma_s^{-1}(Y_{s-1})=\left(\begin{array}{cccc}
\zero & \one &\cdots & \mathbf{p-1}\\
Y_{s-2} & Y_{s-2} &\cdots &Y_{s-2} \\
\end{array}\right).
\end{equation}
Then, we have that

\begin{equation*}
\begin{split}
 &\gamma_s^{-1}(\phi_s(p^{i}u))\\
&=(u_{s-i-1},\stackrel{p^{s-1}}{\dots},u_{s-i-1})+(0,\stackrel{i}{\dots},0,u_0,\dots,u_{s-i-2})\left(\begin{array}{cccc}
\zero & \one &\cdots & \mathbf{p-1}\\
Y_{s-2} & Y_{s-2} &\cdots &Y_{s-2} \\
\end{array}\right)\\
&=(u_{s-i-1},\stackrel{p^{s-1}}{\dots},u_{s-i-1})+(0,\stackrel{i-1}{\dots},0,u_0,\dots,u_{s-i-2})\left(\begin{array}{cccc}Y_{s-2} & Y_{s-2} &\cdots &Y_{s-2}  \\\end{array}\right)\\
&=(\phi_{s-1}(p^{i-1}u),\phi_{s-1}(p^{i-1}u),\dots,\phi_{s-1}(p^{i-1}u))=\Phi_{s-1}(p^{i-1}(u,\stackrel{p}{\dots},u)).   
\end{split}
\end{equation*}
Therefore, $\tau_s(p^{i}u)=\Phi_{s-1}^{-1}(\gamma_s^{-1}(\phi_s(p^{i}u)))=p^{i-1}(u,\stackrel{p}{\dots},u)$, and $(ii)$ holds.
\end{proof}

\begin{proposition}\label{prop:formulaPhis}
Let $s\geq2$ and $\lambda_i\in \{0,1,\dots,p-1\}$, $i\in\{0,\dots,s-1\}$. Then,
\begin{equation}\label{eq:sumPhi8}
\phi_s(\sum_{i=0}^{s-1} \lambda_ip^i)=\gamma_s(\Phi_{s-1}(\sum_{i=0}^{s-1}\tau_s(\lambda_ip^i))).
\end{equation}
\end{proposition}
\begin{proof}
By Lemma \ref{lem:tau}, we know that for all $i\in\{1,\dots,s-1\}$, $\tau_s(p^i)=(p^{i-1},p^{i-1}, \dots, p^{i-1})$ and $\tau_s(1)=(0,p^{s-2},\dots,(p-1)p^{s-2})$.
Then, by Lemma \ref{lema4}, we have that
\begin{equation*}
\gamma_s(\Phi_{s-1}(\sum_{i=0}^{s-1}\tau_s(\lambda_ip^i)))=
\gamma_s(\sum_{i=0}^{s-1}\Phi_{s-1}(\tau_s(\lambda_ip^i))).
\end{equation*}
Moreover, $\gamma_s$ commutes with the addition. Therefore, by applying the definition of the map $\tau_s$ given in (\ref{eq:tau}), we obtain that
$$
\gamma_s(\Phi_{s-1}(\sum_{i=0}^{s-1}\tau_s(\lambda_ip^i)))=
\sum_{i=0}^{s-1}\gamma_s(\Phi_{s-1}(\tau_s(\lambda_ip^i)))=
\sum_{i=0}^{s-1}\phi_s(\lambda_ip^i),
$$
which is equal to $\phi_s(\sum_{i=0}^{s-1} \lambda_ip^i)$ by Lemma \ref{lema4}.
\end{proof}

\begin{corollary}\label{coro:formulaPhis}
Let $s\geq2$ and $\lambda_i\in\{0,1,\dots,p-1\}$, $i\in\{0,\dots,s-1\}$. Then,
\begin{equation}
\tau_{s}(\sum_{i=0}^{s-1} \lambda_ip^i)=\sum_{i=0}^{s-1}\tau_s(\lambda_ip^i).
\end{equation}
\end{corollary}
\begin{proof}
	By Proposition \ref{prop:formulaPhis}, we have that
\[
		 \phi_s(\sum_{i=0}^{s-1} \lambda_ip^i)=\gamma_s(\Phi_{s-1}(\sum_{i=0}^{s-1}\tau_s(\lambda_ip^i))),
		 \]
\noindent and, therefore,
		 \[
		 \Phi_{s-1}^{-1}(\gamma_s^{-1}(\phi_s(\sum_{i=0}^{s-1} \lambda_ip^i)))=\sum_{i=0}^{s-1}\tau_s(\lambda_ip^i),\\
		 \]
		\noindent that is, $\tau_s(\sum_{i=0}^{s-1} \lambda_i2^i)=\sum_{i=0}^{s-1}\tau_s(\lambda_ip^i)$.
\end{proof}

Now, we extend the permutation $\gamma_s\in\mathcal{S}_{p^{s-1}}$ to a permutation $\gamma_s\in\mathcal{S}_{p^{s-1}n}$
such that restricted to each set of $p^{s-1}$ coordinates $\{p^{s-1}i+1,p^{s-1}i+2,\dots,$ $p^{s-1}(i+1)\}$,
$i \in\{0,\dots,n-1\}$, acts as $\gamma_s \in \mathcal{S}_{p^{s-1}}$.
Then, we component-wise extend function $\tau_s$ defined in (\ref{eq:tau}) to $\tau_s:\Z_{p^s}^n\rightarrow
\Z_{p^{s-1}}^{pn}$ and define $\tilde{\tau}_s=\rho^{-1}\circ \tau_s$, where $\rho \in\cS_{pn}$ is defined as follows: for a coordinate $k=jp+i+1 \in \{1,2,\dots, pn\}$, where $i,j  \in \{0,\dots, n-1\}$, $\rho$ moves coordinate $k$ to coordinate $ip+j+1$.
So, we can write $\rho$ as
\begin{equation*}
\footnotesize
\left(
  \begin{array}{ccccccccccccc}
    1 & 2&  \ldots &n & n+1 & n+2 &\ldots & 2n & \ldots & pn-p+1 &pn-p+2 &\ldots &pn \\
    1 & p+1&  \ldots &(n-1)p+1 &2 & p+2 & \ldots & (n-1)p+2 &\ldots & p &p+p & \ldots  &pn\\
  \end{array}
\right).
\end{equation*}

\begin{example}
For $p=3$ and $n=2$,
 \begin{equation*}
 \rho=\left(
   \begin{array}{ccccccccccc}
     1 & 2& &3 & 4& &5& 6 \\
     1 & 4& &2 & 5& &3& 6 \\
   \end{array}
 \right) \in \mathcal{S}_6.
 \end{equation*}
 and for $p=3$ and $n=4$, 
 \begin{equation*}
 \rho=\left(
   \begin{array}{ccccccccccccccccccccccccccccc}
     1 &2 &3 &4 &  &5 &6 &7 &8 &  &9 &10 &11 &12\\
     1 &4 &7 &10 & &2 &5 &8 &11 &  &3 &6 &9 &12\\
 \end{array}
 \right)\in \mathcal{S}_{12}.
  \end{equation*}
\end{example}

\begin{remark}
If $\mathbf{u}=(u_1,u_2,\dots,u_n)\in \Z_{p^s}^n$ and $\tau_s(u_i)=(u_{i,1},u_{i,2},\dots,u_{i,p})$ for all $i \in \{1,\ldots,n\}$, then
$$\tau_s(\mathbf{u})=(u_{1,1},u_{1,2},\dots,u_{1,p},u_{2,1},u_{2,2},\dots,u_{2,p},\dots,u_{n,1},u_{n,2},\dots,u_{n,p}),$$ and
$$\tilde{\tau}_s(\mathbf{u})=(u_{1,1},u_{2,1}\dots,u_{n,1},u_{1,2},u_{2,2}\dots,u_{n,2},\dots, u_{1,p},\\u_{2,p},\dots,u_{n,p}).$$
\end{remark}

\begin{lemma}
   Let $s\geq 2$. Then, $\Phi_s(\mathbf{u})=\gamma_s(\Phi_{s-1}(\rho( \tilde{\tau}_s(\mathbf{u})  )))$
for all $\mathbf{u}\in \Z_{p^s}^n$.
\end{lemma}
\begin{proof}
It follows from $\tilde{\tau}_s(\mathbf{u})=\rho^{-1}(\tau_s(\mathbf{u}))=\rho^{-1}(\Phi_{s-1}^{-1}(\gamma_s^{-1}(\Phi_s(\mathbf{u}))))$.
\end{proof}

Let $\mathbf{w}_i^{(s)}$ be the $i$th row of $A^{t_1,\ldots,t_s}$, $1\leq i\leq t_1+\cdots+t_s$.
By construction, $\mathbf{w}_1^{(s)}=\one$ and  $\ord(\mathbf{w}_i^{(s)})\leq \ord(\mathbf{w}_j^{(s)})$ if $i>j$.
Let $\sigma_i$ be the integer such that $\ord(\mathbf{w}_i^{(s)} )=p^{\sigma_i}$.
Then, $\mathcal{B}^{t_1,\dots,t_s}=\{ p^{q_i}\mathbf{w}_i^{(s)} : 1\leq i\leq t_1+\cdots+t_s,\, 0\leq q_i\leq \sigma_i-1 \}$ is a $p$-basis of $\mathcal{H}^{t_1,\ldots,t_s}$.

\begin{lemma}\label{lem:WsWs1}
Let  $t_s\geq1$. Let $\mathbf{w}_i^{(s)}$ and $\mathbf{w}_i^{(s+1)}$ be the $i$th row of $A^{t_1,\ldots,t_s}$ and $A^{1,t_1-1,t_2,\ldots,t_{s-1},t_s-1}$, respectively. Then,
    $(\mathbf{w}_i^{(s+1)},\mathbf{w}_i^{(s+1)},\dots,\mathbf{w}_i^{(s+1)})=p\mathbf{w}_i^{(s)}$ and $\ord(\mathbf{w}_i^{(s)})=\ord(\mathbf{w}_i^{(s+1)})=p^{\sigma_i}$.
\end{lemma}
\begin{proof}
Consider $A^{t_1,\dots,t_s}$ with $t_s\geq1$, and $\mathbf{w}_i^{(s)}$ its $i$th row for $i\in\{1,\dots,t_1+\cdots+t_s\}$. Then, the matrix over $\Z_{p^{s+1}}$
$$
\left(\begin{array}{c}
\mathbf{w}_1^{(s)}\\
p\mathbf{w}_2^{(s)}\\
\vdots\\
p\mathbf{w}_{t_1+\cdots+t_s}^{(s)}\\
\end{array}\right)
$$
is, by definition, $A^{1,t_1-1,t_2,\dots,t_{s-1},t_s}$. Moreover, by construction we have that $A^{1,t_1-1,t_2,\dots,t_{s-1},t_s}$ is the matrix
$$
\left(\begin{array}{cccc}
A^{1,t_1-1,t_2,\dots,t_{s-1},t_s-1} & A^{1,t_1-1,t_2,\dots,t_{s-1},t_s-1}&\cdots &A^{1,t_1-1,t_2,\dots,t_{s-1},t_s-1}\\
 p^s\cdot \zero & p^s\cdot\mathbf{1}  &\cdots &p^s \cdot \mathbf{(p-1)}\\
\end{array}\right).
$$
Therefore, if $\mathbf{w}_i^{(s+1)}$ is the $i$th row of $A^{1,t_1-1,t_2,\dots,t_{s-1},t_s-1}$ for $i\in\{2,\dots,t_1+t_2+\dots+t_s-1\}$, we have that $(\mathbf{w}_i^{(s+1)},\mathbf{w}_i^{(s+1)},\dots,\mathbf{w}_i^{(s+1)})=p\mathbf{w}_i^{(s)}$ and $\ord(\mathbf{w}_i^{(s)})=\ord(\mathbf{w}_i^{(s+1)})=p^{\sigma_i}$.
\end{proof}

\begin{example} \label{ex:2bases}
Let ${\mathcal H}^{2,1}$ and ${\mathcal H}^{1,1,0}$ be the $\Z_9$-additive and $\Z_{27}$-additive GH codes, which are generated by

$$
\footnotesize
A^{2,1}=\left(\begin{array}{c}
\mathbf{w}_1^{(2)}\\
\mathbf{w}_2^{(2)}\\
\mathbf{w}_3^{(2)}
\end{array}\right)
=\left(\begin{array}{ccc}
\one & \one & \one\\
v & v & v \\
\zero & \mathbf{3} & \mathbf{9}
\end{array}\right),
$$
where $v=(0,1,2,3,4,5,6,7,8)$, and
$$
A^{1,1,0}=\left(\begin{array}{c}
\mathbf{w}_1^{(3)}\\
\mathbf{w}_2^{(3)}
\end{array}\right)
=\left(\begin{array}{ccccccccc}
1&1&1&1&1&1&1&1&1 \\
0&3&6&9&12&15&18&21&24 \\
\end{array}\right),
$$
respectively.
The corresponding $3$-bases are
\begin{align*}
\mathcal{B}^{2,1}=\{&\mathbf{w}_1^{(2)},3\mathbf{w}_1^{(2)},\mathbf{w}_2^{(2)},3\mathbf{w}_2^{(2)},\mathbf{w}_3^{(2)}\} \\
=\{&\mathbf{1},\mathbf{3},(0,1,2,3,4,5,6,7,8,\stackrel{3}{\dots},0,1,2,3,4,5,6,7,8),\\
&(0,3,6,9,12,15,18,21,24,\stackrel{3}{\dots},0,3,6,9,12,15,18,21,24)\},~ \mbox{and} \\
\mathcal{B}^{1,1,0}=\{&\mathbf{w}_1^{(3)},3\mathbf{w}_1^{(3)},3^2\mathbf{w}_1^{(3)},\mathbf{w}_2^{(3)},3\mathbf{w}_2^{(3)} \}\\
=\{&\mathbf{1},\mathbf{3},\mathbf{9},(0,3,6,9,12,15,18,21,24),(0,9,18,0,9,18,0,9,18)\}.\\
\end{align*}
\end{example}

\begin{proposition} \label{prop:basis}
Let  $t_s\geq1$, and $\mathcal{H}^{t_1,\ldots,t_s}$ and $\mathcal{H}^{1,t_1-1,t_2,\ldots,t_{s-1},t_s-1}$ be the $\Z_{p^s}$-additive and $\Z_{p^{s+1}}$-additive GH codes with generator matrices $A^{t_1,\ldots,t_s}$ and $A^{1,t_1-1,t_2,\ldots,t_{s-1},t_s-1}$, respectively. Let $\mathbf{w}_i^{(s)}$ and $\mathbf{w}_i^{(s+1)}$ be the $i$th row of $A^{t_1,\ldots,t_s}$ and $A^{1,t_1-1,t_2,\ldots,t_{s-1},t_s-1}$, respectively. Then, we have that
\begin{enumerate}[label=(\roman*)]
\item $\tilde{\tau}_{s+1}(p^{q_i}\mathbf{w}^{(s+1)}_i)=p^{q_i}\mathbf{w}^{(s)}_i$,  for all $i \in \{2,\dots,t_1+\cdots+t_s-1\}$ and $q_i\in\{0,\dots,\sigma_i-1\}$, where $\ord(\mathbf{w}_i^{(s)})=p^{\sigma_i}$;
\item $\tilde{\tau}_{s+1}(p^{j+1}\mathbf{w}^{(s+1)}_1)=p^{j}\mathbf{w}^{(s)}_1$,  for all $j\in\{0,\dots,s-1\}$;
\item $\tilde{\tau}_{s+1}(\mathbf{w}_1^{(s+1)})=\mathbf{w}^{(s)}_{t_1+\cdots +t_{s}}.$
\end{enumerate}
\end{proposition}
\begin{proof}
By Lemma \ref{lem:WsWs1}, we have that $(\mathbf{w}_i^{(s+1)},\mathbf{w}_i^{(s+1)},\dots,\mathbf{w}_i^{(s+1)})=p\mathbf{w}_i^{(s)}$ and $\ord(\mathbf{w}_i^{(s)})=\ord(\mathbf{w}_i^{(s+1)})=p^{\sigma_i}$. Let $\mathbf{v}_i^{(s+1)}$ be the vector over $\Z_{p^{s+1}}$ such that $\mathbf{w}_i^{(s+1)}=p\mathbf{v}_i^{(s+1)}$ and $\mathbf{w}_i^{(s)}=(\mathbf{v}_i^{(s+1)},\mathbf{v}_i^{(s+1)},\dots,\mathbf{v}_i^{(s+1)})$. Let $(\mathbf{v}_i^{(s+1)})_j$ be the $j$th coordinate of $\mathbf{v}_i^{(s+1)}$. By the definition of $\tilde{\tau}_{s+1}$ and Lemma \ref{lem:tau},  for $q_i\in\{0,\dots,\sigma_i-1\}$, we have that
\begin{equation*}
\begin{split}
&\tilde{\tau}_{s+1}(p^{q_i}\mathbf{w}^{(s+1)}_i)=\rho^{-1}(\tau_{s+1}(p^{q_i}\mathbf{w}^{(s+1)}_i))=\rho^{-1}(\tau_{s+1}(p^{q_i+1}\mathbf{v}_i^{(s+1)}))\\
&=\rho^{-1}(p^{q_i}((\mathbf{v}_i^{(s+1)})_1,\dots,(\mathbf{v}_i^{(s+1)})_1,\dots,(\mathbf{v}_i^{(s+1)})_n,\dots,(\mathbf{v}_i^{(s+1)})_n))\\
&=p^{q_i}(\mathbf{v}_i^{(s+1)},\mathbf{v}_i^{(s+1)},\dots,\mathbf{v}_i^{(s+1)})=p^{q_i}\mathbf{w}_i^{(s)},
\end{split}
\end{equation*}
and $(i)$ holds.

Since $\mathbf{w}_1^{(s)}=(\mathbf{w}_1^{(s+1)},\dots,\mathbf{w}_1^{(s+1)})=\one$ and $\mathbf{w}_{t_1+\cdots+t_s}^{(s)}=(p^{s-1} \cdot \zero,p^{s-1} \cdot \mathbf{1},\dots,p^{s-1}\cdot \mathbf{(p-1)})$, then the equalities in items $(ii)$ and $(iii)$ hold by the definition of $\tilde{\tau}_{s+1}$ and Lemma \ref{lem:tau}.
\end{proof}

Note that, from the previous proposition, we have that $\tilde{\tau}_{s+1}$ is a bijection between the $p$-bases, $\mathcal{B}^{t_1,\dots,t_s}$ and $\mathcal{B}^{1,t_1-1,\dots,t_{s-1},t_s-1}$.

\begin{example}
Let ${\mathcal H}^{2,1}$  and  ${\mathcal H}^{1,1,0}$ be the same codes considered
in Example \ref{ex:2bases}. The length of ${\mathcal H}^{1,1,0}$ is $n=9$.
Then, the extension of $\gamma_3= (2,4)(3,7)(6,8)$ $ \in \cS_9$ is
$\gamma_3=(2,4)(3,7)(6,8)(11,13)(12,16)(15,17)(20,22)\\(21,25)(24,26)(29,31)(30,34)(33,35)(38,40)(39,43)(42,44)(47,49)(48,52)\\(51,53)(56,58)(57,61)(60,62)(65,67)(66,70)(69,71)(74,76)(75,79)(78,80)\in \cS_{81}$, and

\begin{equation}
\tiny
\rho=\left(
  \begin{array}{cccccccccccccccccccccccccccccc}
    1 & 2& 3 & 4 & 5 & 6 & 7 & 8 &9    &10 &11 &12 &13 &14 &15 &16 &17 &18    &19 &20 &21 &22 &23 &24 &25 &26 &27 \\
    1 &4 &7  &10  &13  &16  &19     &22     &25              &2  &5  &8  &11   &14    &17   &20  &23    &26               &3 &6 &9  &12  &15  &18  &21   &24   &27    \\
  \end{array}
\right)  \in \cS_{27}.
\end{equation}
In this case, we have that
$$
\footnotesize
\arraycolsep=1pt\def\arraystretch{}
\begin{array}{lccclcc}
\Phi_3(1,1,1,1,1,1,1,1,1)&=&\gamma_3(\Phi_2(0,3,6,  \stackrel{9}{\dots},0,3,6))&=&\gamma_3(\Phi_2(\rho(\mathbf{0},\mathbf{3},\mathbf{6}))), \\
\Phi_3(3,3,3,3,3,3,3,3,3)&=&\gamma_3(\Phi_2(1,1,1,\stackrel{9}{\dots},1,1,1))&=&\gamma_3(\Phi_2(\rho(1,1,1,\stackrel{9}{\dots},1,1,1))), \\
\Phi_3(9,9,9,9,9,9,9,9,9)&=&\gamma_3(\Phi_2(3,3,3,\stackrel{9}{\dots},3,3,3))&=&\gamma_3(\Phi_2(\rho(3,3,3,\stackrel{9}{\dots},3,3,3))), \\
\Phi_3(0,3,6,9,12,15,18,21,24)&=&\gamma_3(\Phi_2(0,0,0,1,1,1,\dots,8,8,8))&=&\gamma_3(\Phi_2(\rho(u,\stackrel{3}{\dots},u)), \\
\Phi_3(0,9,18,0,9,18,0,9,18)&=&\gamma_3(\Phi_2(v, \stackrel{3}{\dots}, v))&=&\gamma_3(\Phi_2(\rho(0,3,6, \stackrel{9}{\dots},0,3,6))),
\end{array}
$$ where $u=(0,1,2,3,4,5,6,7,8)$ and $v=(0,0,0,3,3,3,6,6,6).$

Since $\Phi_3(\mathbf{u})=\gamma_3(\Phi_{2}(\rho( \tilde{\tau}_3(\mathbf{u})  )))$ for all $\mathbf{u}\in \Z_{27}^9$,
the map $\tilde{\tau}_3$ sends the elements of the $3$-basis $\mathcal{B}^{1,1,0}$ into the elements
of the $3$-basis $\mathcal{B}^{2,1}$.
That is, as it is shown in Proposition~\ref{prop:basis},
\begin{equation*}
\arraycolsep=1pt\def\arraystretch{}
\begin{array}{lcccccl}
\tilde{\tau}_3(\mathbf{w}_1^{(3)}) & = & \tilde{\tau}_3(1,1,1,1,1,1,1,1,1) & = & (\mathbf{0},\mathbf{3},\mathbf{6}) & = & \mathbf{w}_3^{(2)},\\
\tilde{\tau}_3(3\mathbf{w}_1^{(3)}) & = & \tilde{\tau}_3(3,3,3,3,3,3,3,3,3) & = & (1,1,1,\stackrel{9}{\dots},1,1,1) & = & \mathbf{w}_1^{(2)},\\
\tilde{\tau}_3(9\mathbf{w}_1^{(3)}) & = & \tilde{\tau}_3(9,9,9,9,9,9,9,9,9) & = & (3,3,3,\stackrel{9}{\dots},3,3,3) & = & 3\mathbf{w}_1^{(2)},\\
\tilde{\tau}_3(\mathbf{w}_2^{(3)}) & = & \tilde{\tau}_3(0,3,6,9,12,15,18,21,24) & = & (u,\stackrel{3}{\dots}, u) & = & \mathbf{w}_2^{(2)},\\
\tilde{\tau}_3(3\mathbf{w}_2^{(3)}) & = & \tilde{\tau}_3(0,9,18,0,9,18,0,9,18) & = & (0,3,6,\stackrel{9}{\dots},0,3,6) & = & 3\mathbf{w}_2^{(2)},\\
\end{array}
\end{equation*} 
so $\tilde{\tau}_3$ is a bijection between both $3$-bases.
\end{example}

\begin{lemma}\label{lemma:basisEquivcodes}
Let $\mathcal{H}_s=\mathcal{H}^{t_1,\ldots,t_s}$ be a $\Z_{p^s}$-additive GH code with $t_s\geq1$, and $\mathcal{H}_{s+1}=\mathcal{H}^{1,t_1-1,t_2,\ldots,t_{s-1},t_s-1}$ be a $\Z_{p^{s+1}}$-additive GH code.
Then, $H_s=\Phi_s(\mathcal{H}_s)$ is permutation equivalent to $H_{s+1}=\Phi_{s+1}(\mathcal{H}_{s+1})$.
\end{lemma}
\begin{proof}
Let $t$ be the integer such that $H_s$ and $H_{s+1}$ are of length $p^t$.
Let $\mathcal{B}_s=\{\bv_1^{(s)},\dots,\bv_{t+1}^{(s)}\}$ and $\mathcal{B}_{s+1}=\{\bv_1^{(s+1)},\dots,\bv_{t+1}^{(s+1)}\}$ be the $p$-basis of $\mathcal{H}_s$ and  $\mathcal{H}_{s+1}$, respectively.
By Proposition \ref{prop:basis}, $\tilde{\tau}_{s+1}$ is a bijection between $\mathcal{B}_s$ and $\mathcal{B}_{s+1}$. By the definition of
$\tilde{\tau}_{s+1}$ and Corollary \ref{coro:formulaPhis}, we have that $\tilde{\tau}_{s+1}$ commutes with the addition, so $\tilde{\tau}_{s+1}(\mathcal{H}_{s+1})=\mathcal{H}_{s}$.

  Let $\rho_*\in\mathcal{S}_{p^t}$ be a permutation such that $\Phi_s(\rho(\mathbf{u}))=\rho_*(\Phi_s(\mathbf{u}))$ for all $\mathbf{u}\in\mathcal{H}_{s}$. Since $\mathcal{H}_s=\tilde{\tau}_{s+1}(\mathcal{H}_{s+1})=\rho^{-1}(\Phi_{s}^{-1}(\gamma^{-1}_{s+1}(\Phi_{s+1}(\mathcal{H}_{s+1}))))$, we have that $\Phi_{s}(\mathcal{H}_s)=\rho_*^{-1}(\gamma^{-1}_{s+1}(\Phi_{s+1}(\mathcal{H}_{s+1})))$. Therefore, we obtain $H_s=(\gamma_{s+1}\circ\rho_*)^{-1}(H_{s+1})$, where $\gamma_{s+1}\circ\rho_*\in\mathcal{S}_{p^t}$.
\end{proof}

The following theorem determines which $\Z_{p^{s'}}$-linear GH codes
are equivalent to a given $\Z_{p^s}$-linear Hadamard code
$H^{t_1,\dots,t_s}$. We denote by $\zero^j$ the all-zero vector of length
$j$. Let $\sigma$ be the integer such that $\ord(\mathbf{w}_2^{(s)})=p^{s+1-\sigma}$. Note that $\sigma =1$ if and only if $t_1\geq 2$. Moreover, since $\sigma_2$ is the integer such that $\ord(\mathbf{w}_2^{(s)})=p^{\sigma_2}$, we have that $\sigma=s+1-\sigma_2$.

\begin{theorem}\label{theo:equi}
Let $H^{t_1,\dots,t_s}$ be a $\Z_{2^s}$-linear GH code with $t_s\geq 1$.
Then, $H^{t_1,\dots,t_s}$  is equivalent to the
$\Z_{p^{s+\ell}}$-linear Hadamard code
$H^{1,\zero^{\ell-1},t_1-1,t_2,\dots,t_{s-1},t_s-\ell}$, for all
$\ell\in\{1,\dots,t_s\}$.
\end{theorem}
\begin{proof}
Consider $\mathcal{H}_0=\mathcal{H}^{t_1,\dots,t_s}$ and $\mathcal{H}_\ell=\mathcal{H}^{1,\zero^{\ell-1},t_1-1,t_2,\dots,t_{s-1},t_s-\ell}$ for $\ell\in\{1,\dots,t_s\}$. By Lemma \ref{lemma:basisEquivcodes}, we have that $H_i=\Phi(\mathcal{H}_i)$ is permutation equivalent to $H_{i+1}=\Phi(\mathcal{H}_{i+1})$ for all $i\in\{0,\dots,\ell-1\}$ and $\ell\in\{1,\dots,t_s\}$. Therefore, we have that $H_0$ and $H_{\ell}$ are permutation equivalent for all $\ell\in\{1,\dots,t_s\}$.
\end{proof}

 \medskip
Let $t_1, t_2,\dots,t_s$ be nonnegative integers with $t_1\geq 2$. Let $C_H(t_1,\dots,t_s)=[H_1=H^{t_1,\dots,t_s},H_2,\dots, H_{t_s+1}]$ be the sequence of all $\Z_{p^{s'}}$-linear GH codes of length $p^t$, where $t=$ $\left(\sum_{i=1}^{s'}(s'-i+1)\cdot t_i\right)-1$, that are permutation equivalent to $H^{t_1,\dots,t_s}$ by Theorem \ref{theo:equi}.
We denote by $C_H(t_1,\dots,t_s)[i]$ the $i$th code $H_i$ in the sequence, for $1\leq i\leq t_s+1$. We consider that the order of the codes in $C_H(t_1,\dots,t_s)$ is the following:
\begin{equation}\label{eq:defCH}
	C_H(t_1,\dots,t_s)[i]=
		\begin{cases}
	H^{t_1,\dots,t_{s}},& \text{if } i= 1,\\
	H^{1,\zero^{i-2},t_1-1,t_2,\dots,t_{s-1},t_s-i+1},              & \text{otherwise.}
\end{cases}
\end{equation}
 We refer to $C_H(t_1,\dots,t_s)$ as the {\it chain of equivalences} of $H^{t_1,\dots,t_s}$. Note that if $t_s=0$, then $C_H(t_1,\dots,t_s)=[H^{t_1,\dots,t_s}]$.

\begin{corollary}\label{lemma:CH}
	Let $t_1, t_2,\dots,t_s$ be nonnegative integers with $t_1\geq 2$. Then,
	$$|C_H(t_1,\dots,t_s)|=t_s+1.$$
\end{corollary}

\begin{example}\label{examples:chain1}
 Let $p=3$.	Then the chain of equivalences of $H^{3,3}$ is the sequence $C_H(3,3)=[H^{3,3}, H^{1,2,2}, H^{1,0,2,1}, H^{1,0,0,2,0}]$ and contains exactly $t_2+1=4$ codes. Note that this sequence contains all the codes of length $3^8$ having the pair $(r,k)=(14,6)$ in Table \ref{table:TypesB2}. In the  table \ref{table:Types}, we can see that there is only one code of length $3^7$ having the pair $(r,k)=(25,3)$, named $H^{2,1,0}$. The chain of equivalences of this code is $C_H(2,1,0)=[H^{2,1,0}]$, which contains just this code, since $t_3+1=1$.
\end{example}

\begin{proposition}\label{lemma:pos-i-in-CH}
	Let $t_1,t_2,\dots,t_s$ be nonnegative integers with $t_1\geq 2$. Then, the  $\Z_{p^{s'}}$-linear Hadamard code
	$H^{t'_1,\dots,t'_{s'}}= C_H(t_1,\dots,t_s)[i]$, $1\leq i\leq t_s+1$, satisfies
	\begin{enumerate}[label=(\roman*)]
		\item $s'=s+i-1$,
		\item $\sigma'=i$,
		\item $t'_{s'}=t_s-i+1$,
		\item $
		(t'_1,\dots,t'_{s'})=
		\begin{cases}
		(t_1,\dots,t_{s}),& \text{if } i= 1,\\
		(1,\zero^{i-2},t_1-1,t_2,\dots,t_{s-1},t_s-i+1),              & \text{otherwise.}
		\end{cases}
		$	
	\end{enumerate}
\end{proposition}

\begin{proof}
Straightforward from Theorem \ref{theo:equi} and the definition of the
chain of equivalences $C_H(t_1,\dots,t_s)$.
\end{proof}

Note that the value of $s'$ is different for every code $H^{t'_1,\dots,t'_{s'}}$ 
belonging to  the same chain of equivalences $C_H(t_1,\dots,t_s)$.

\begin{corollary}\label{coro:One-ts0}
	Let $t_1, t_2,\dots,t_s$ be nonnegative integers with $t_1\geq 2$.
	\begin{enumerate}[label=(\roman*)]
		\item 	The $\Z_{p^{s'}}$-linear Hadamard code $H^{t'_1,\dots,t'_{s'}}= C_H(t_1,\dots,t_s)[1]$ is the only one in $C_H(t_1,\dots,t_s)$ with $\sigma'=1$, that is, the only one with $t'_1\geq 2$.
		\item	The $\Z_{p^{s'}}$-linear Hadamard code $H^{t'_1,\dots,t'_{s'}}= C_H(t_1,\dots,t_s)[t_s+1]$ is the only one in $C_H(t_1,\dots,t_s)$ with $t'_{s'}=0$.	
	\end{enumerate}
\end{corollary}

Now, given any $\Z_{p^s}$-linear Hadamard code $H^{t_1,\dots,t_s}$, we determine the chain of equivalences containing this code, as well as its position in the sequence. Therefore, note that indeed we prove that any code $H^{t_1,\dots,t_s}$ (with $t_1\geq 1$) belongs to a unique chain of equivalences
$C_H(t'_1,\dots,t'_{s'})$ with $t'_1 \geq 2$.

\begin{proposition}\label{lemma:codeInChain}
Let $t_1, t_2,\dots,t_s$ be nonnegative integers with $t_1\geq 1$. The $\Z_{p^{s}}$-linear Hadamard code $H^{t_1,\dots,t_{s}}$ belongs to
an unique chain of equivalences.
If $t_1\geq 2$, then $\sigma=1$ and $H^{t_1,\dots,t_s}=C_H(t_1,\dots,t_s)[1]$. Otherwise, if $t_1=1$, then  $\sigma >1$ and $H^{t_1,\dots,t_s}=C_H(t'_1,\dots,t'_{s'})[\sigma]$, where $(t'_1,\dots,t'_{s'})= (t_\sigma+1,t_{\sigma+1},\dots,t_{s-1},t_s+\sigma-1)$ and $s'=s-\sigma+1$.
\end{proposition}
\begin{proof}
If $t_1\geq 2$, it is clear by the definition of chain of equivalences and Corollary \ref{coro:One-ts0}.
If $t_1=1$, then $\sigma >1$. By Proposition \ref{lemma:pos-i-in-CH}, since $H^{t_1,\dots,t_s}=C_H(t'_1,\dots,t'_{s'})[\sigma]$, we have that $(t_1,\dots,t_s)=(1,\zero^{\sigma-2},t'_1-1, t'_2, \ldots,t'_{s'-1}, t'_{s'}-\sigma+1)$. Therefore, $t_\sigma=t'_1-1$, $t_{\sigma+1}=t'_2$, $\ldots$, $t_{s-1}=t'_{s'-1}$, $t_s=t'_{s'}-\sigma+1$, and the result follows.
\end{proof}

\begin{corollary}\label{coro:count-CH}
Let $H^{t_1,\dots,t_{s}}$ be a $\Z_{p^{s}}$-linear Hadamard code and $C_H(t'_1,\dots,t'_{s'})$ the chain of equivalences such that $H^{t_1,\dots,t_{s}}=C_H(t'_1,\dots,t'_{s'})[\sigma]$. Then,
$$	|C_H(t'_1,\dots,t'_{s'})|=t_s+\sigma.$$
\end{corollary}
\begin{proof}
If $t_1\geq 2$, then $H^{t_1,\dots,t_s}=C_H(t_1,\dots,t_{s})[1]$ by Proposition \ref{lemma:codeInChain}. By Corollary \ref{lemma:CH}, $|C_H(t_1,\dots,t_{s})|=t_s+1=t_s+\sigma$. Otherwise, if $t_1=1$, then $H^{t_1,\dots,t_s}=C_H(t_\sigma+1,t_{\sigma+1},\dots,$ $t_{s-1},t_s+\sigma-1)[\sigma]$ by Proposition \ref{lemma:codeInChain}. Finally, $|C_H(t_\sigma+1,t_{\sigma+1},\dots,$ $t_{s-1},t_s+\sigma-1)|=t_s+\sigma-1+1=t_s+\sigma$ by Corollary \ref{lemma:CH}.
\end{proof}

\begin{example}\label{examples:chainpos}
The $\Z_{3^4}$-linear GH code $H^{1,0,2,1}$ has $t_1=1$, $\sigma=3$ and $s=4$.
By Proposition \ref{lemma:codeInChain}, since $\sigma=3$ and $s'=s-\sigma+1=2$, this code is placed in the third position of the chain
of equivalences $C_H(t'_1,t'_2)$, where $t'_1=t_3+1=2+1=3$ and $t'_2=t_4+\sigma-1=1+3-1=3$.
Therefore, $H^{1,0,2,1}=C_H(3,3)[3]$. By Corollary \ref{coro:count-CH}, $C_H(3,3)$ contains exactly $t_4+\sigma=1+3=4$ codes, which are
the ones described in Example \ref{examples:chain1}.
\end{example}

If $H^{t_1,\dots,t_s}$ is a $\Z_{3^s}$-linear GH code of length $3^t$ with $t_1\geq 2$ and $t_s=0$, then $|C_H(t_1,\dots,t_s)|=1$ by Corollary \ref{coro:count-CH}. In this case, from Tables \ref{table:Types} and \ref{table:TypesB2},
we can see that $H^{t_1,\dots,t_s}$ is not equivalent to any other code $H^{t'_1,\dots,t'_{s'}}$ of the same length $3^t$, for $t\leq 10$.
We conjecture that this is true for any $t\geq 11$. The values of $(t_1,\dots,t_s)$ for which the codes $H^{t_1,\dots,t_s}$ are not equivalent to any other such code of the same length $3^t$, for $t\leq 10$, can be found in Table \ref{table:Types2}.

\begin{table}[h]
	\centering
	\begin{tabular}{|c||l|}
		\hline
		$t=5$ & (3,0), (2,0,0)\\
		\hline
		$t=7$ & (4,0), (2,1,0), (2,0,0,0)\\
		\hline
		$t=8$ & (3,0,0)\\
		\hline
		$t=9$ & (5,0), (2,2,0), (2,0,1,0), (2,0,0,0,0)\\
		\hline
		$t=10$ & (3,1,0), (2,1,0,0)\\
		\hline
	\end{tabular}
	\caption{Type of all $\Z_{3^s}$-linear GH codes of length $3^t$ with $\sigma=1$ and $t_s=0$ for $t \leq 10$.}
	\label{table:Types2}
\end{table}

From Tables \ref{table:Types} and \ref{table:TypesB2}, we can also see that the $\Z_{3^s}$-linear GH codes
of length $3^t$ with $t\leq10$ having the same values $(r,k)$ are the ones which are equivalent by Theorem \ref{theo:equi}.
We conjecture that this is true for any $t\geq 11$.

\section{Improvement of the partial classification}
\label{sec:Classification}

In this section, we improve some results on the classification of the $\Z_{p^s}$-linear GH codes of length $p^t$, once $t$ is fixed. More precisely, we improve the upper bounds of $\cA_{t,p}$ given by Theorem \ref{theo:CardinalFirstBound}  and determine the exact value of $\cA_{t,p}$ for $t\leq 10$ by using the equivalence results established in Section \ref{sec:relations}.

Next, we prove two corollaries of Theorem \ref{theo:equi}, which allow us to improve the known upper bounds for $\cA_{t,p}$.

\begin{corollary} \label{coro:t1}
	Let $H^{t_1,\dots,t_s}$ be a $\Z_{p^s}$-linear GH code. Then, $H^{t_1,\dots,t_s}$ is equivalent to $t_s + \sigma$ $\Z_{p^{s'}}$-linear GH codes for $s'\in \{s+1-\sigma,\dots,s+t_s\}$. Among them, there is exactly one $H^{t'_1,\dots,t'_{s'}}$ with $t'_1\geq 2$, and there is exactly one $H^{t'_1,\dots,t'_{s'}}$ with $t'_{s'}=0$.
\end{corollary}
\begin{proof}
	The code $H^{t_1,\dots,t_s}$ belongs to a chain of equivalences $C$, which can be determined by Proposition \ref{lemma:codeInChain}. We have that  $H^{t_1,\dots,t_s}$ is equivalent to any code in $C$ by Theorem \ref{theo:equi}, and the number of codes in $C$ is $t_s+\sigma$ by Corollary \ref{coro:count-CH}. By Proposition \ref{lemma:codeInChain}, the first code in $C$ has $s'=s-\sigma+1$. By Proposition \ref{lemma:pos-i-in-CH}, the $i$th code
$H^{t'_1,\dots,t'_{s'}}=C[i]$ has $s'=s-\sigma+i$ for $i\in\{1,\dots, t_s+\sigma\}$.
Therefore, $s'\in \{ s-\sigma+1, \dots, s+t_s \}$.
Finally, by Corollary \ref{coro:One-ts0}, $C[1]$ is the only code in $C$ with $t'_1\geq 2$ and $C[t_s+\sigma]$ is the only code in $C$ with $t'_{s'}=0$.	
\end{proof}

From Corollary \ref{coro:t1}, in order to determine the number of nonequivalent $\Z_{p^s}$-linear GH codes of length $p^t$,
$\cA_{t,p}$, we just have to consider one code out of the $t_s+\sigma$  codes that are equivalent. For example, we can consider the one with $t_1\geq 2$.

\begin{corollary}\label{cor:sin}
Let $H$ be a nonlinear $\Z_{p^s}$-linear GH code of length $p^t$. If
$s\in\{\lfloor(t+1)/2\rfloor+1,\dots,t+1\}$, then there is an equivalent
$\Z_{p^{s'}}$-linear Hadamard code of length $p^t$ with
$s'\in\{2,\dots,\lfloor(t+1)/2\rfloor\}$.
\end{corollary}
\begin{proof}
Let $H^{t_1,\dots,t_s}$ be a $\Z_{p^s}$-linear Hadamard code with $s\in\{\lfloor(t+1)/2\rfloor+1,\dots,t+1\}$. Since $\sum_{i=1}^s(s+1-i)t_i=t+1$, then $t_1=1$ and we have that $\sigma>1$. Therefore, by Proposition \ref{lemma:codeInChain}, $H^{t_1,\dots,t_s}$ is permutation equivalent to the $\Z_{p^{s-\sigma+1}}$-linear Hadamard code $H=H^{t_\sigma+1,t_{\sigma+1},\dots,t_{s-1},t_s+\sigma-1}$.

Now, we just need to see that $s-\sigma+1<\lfloor(t+1)/2\rfloor$. Since the length of $H$ is $p^t$, we have that $t+1=(s-\sigma+1)(t_\sigma+1)+\sum_{i=2}^{s-\sigma+1}(s-\sigma+2-i)t_{\sigma-1+i}+\sigma-1$. Therefore, $(s-\sigma+1)(t_\sigma+1)\leq t+1$ and $s-\sigma+1\leq({t+1})/({t_\sigma+1})$. By the definition of $t_\sigma$, we know that $t_\sigma\geq1$, so $s-\sigma+1\leq \lfloor (t+1)/2\rfloor$.
\end{proof}

Note that we can focus on the $\Z_{p^s}$-linear GH codes of length $p^t$ with $s\in \{2,\dots,\lfloor (t+1)/2\rfloor\}$ by Corollary \ref{cor:sin}, and  we can restrict ourselves to the codes having $t_1\geq 2$ by Corollary \ref{coro:t1}.
With this on mind, in order to classify all such codes for a given $t\geq 3$, we define $\tilde{X}_{t,s,p}=|\{ (t_1,\ldots,t_s)\in \N^s :  t+1=\sum_{i=1}^{s}(s-i+1) t_i, \ t_1\geq 2 \}|$ for $s\in\{2,\dots,\lfloor (t+1)/2\rfloor\}$. 

\begin{theorem} \label{teo:numNonEquiv100-p}
Let $\cA_{t,p}$ be the number of nonequivalent $\Z_{p^s}$-linear GH codes of length $p^t$ with $t\geq 3$ and $p\geq 3$ prime. Then,
\begin{equation}\label{eq:BoundAtNewX-p}
\cA_{t,p} \leq1+ \sum_{s=2}^{\lfloor\frac{t+1}{2}\rfloor} \tilde{X}_{t,s,p}
\end{equation}
and
\begin{equation}\label{eq:BoundAtNew-p}
\cA_{t,p} \leq1+ \sum_{s=2}^{\lfloor\frac{t+1}{2}\rfloor} ({\cA}_{t,s,p}-1).
\end{equation}
Moreover, for $p\geq 3$ prime and  $3 \leq t \leq 10$, the upper bound (\ref{eq:BoundAtNewX-p}) is tight.
\end{theorem}
\begin{proof}
It is proven that the codes $H^{1,0,\dots,0,t_s}$ with $s\geq2$ and $t_s\geq0$, are the only $\Z_{p^s}$-linear GH codes which are linear. Note that they are not included in the definition of $\tilde{X}_{t,s,p}$ for $s\in \{2,\dots,\lfloor (t+1)/2\rfloor\}$. 
Therefore, the new upper bounds (\ref{eq:BoundAtNewX-p}) and (\ref{eq:BoundAtNew-p}) follow by Corollaries \ref{coro:t1} and \ref{cor:sin}, after adding $1$ to take into account the linear code.

In Table \ref{table:ClassificationEquiv}, for any odd prime $p$ and $3\leq t \leq 10$, these new upper bounds together with previous bounds
are shown. Note that the lower bound $(r,k)$ coincides with the upper bound (\ref{eq:BoundAtNewX-p}) for $t\leq10$, so this upper bound is tight for $t\leq10$.
\end{proof}


\begin{table}[h]
\centering
\begin{tabular}{|c||c|c|c|c|c|c|c|c|}
\cline{0-8}
$t$ & 3 & 4 & 5 & 6 & 7 & 8 & 9 & 10 \\ \cline{0-8}
previous lower bound $(r,k)$ & {\bf 2}  & {\bf 2}  & {\bf 4}  & {\bf 4} & {\bf 7} & {\bf 8} & {\bf 12} & {\bf 14} \\\cline{0-8}
new upper bound (\ref{eq:BoundAtNewX-p}) & {\bf 2}  & {\bf 2}  & {\bf 4}  & {\bf 4} & {\bf 7} & {\bf 8} & {\bf 12} & {\bf 14}  \\\cline{0-8}
new upper bound (\ref{eq:BoundAtNew-p})  & 2  & 2  & 5  & 6 & 11 & 15 & 26 & 33    \\\cline{0-8}
previous upper bound  & 2  & 2  & 6  & 9 & 15 & 22 & 33 & 46  \\\cline{0-8}
\end{tabular}
\caption{Bounds for the number $\cA_{t,p}$ of nonequivalent $\Z_{p^s}$-linear GH codes of length $p^t$ for $3\leq t \leq 10$.}
\label{table:ClassificationEquiv}
\end{table}

This last result improves the partial classification. Actually, by definition, we have that $\tilde{X}_{t,s,p}\leq X_{t,s,p}-1$, so the upper bound (\ref{eq:BoundAtNewX-p})
is clearly better than (\ref{eq:CardinalFirstBound}). It is also clear that the upper bound (\ref{eq:BoundAtNew-p}) is better than (\ref{eq:CardinalFirstBoundX-p}) since
there are fewer addends. Therefore, both new upper bounds improve the previous known upper bounds.
Recall that $\cA_{t,s,p}=X_{t,s,p}$ for any $3\leq t\leq 10$ and $2\leq s\leq t-2$. If this equation is also true for any $t\geq 11$, then upper bounds $(\ref{eq:CardinalFirstBound-p})$ and $(\ref{eq:CardinalFirstBoundX-p})$
coincide. Moreover, upper bound (\ref{eq:BoundAtNewX-p}) would always be better than (\ref{eq:BoundAtNew-p}) since $\tilde{X}_{t,s,p}\leq X_{t,s,p}-1 = \cA_{t,s,p}-1$.

\end{document}